\documentclass[pra,amsmath,amssymb,amsfonts,twocolumn,superscriptaddress,nofootinbib]{revtex4-1}
\usepackage{bm,graphicx,mathrsfs}

\usepackage{graphicx}
\usepackage{epsfig}
\usepackage{amsmath,bbm}
\usepackage{amsfonts,amssymb}
\usepackage{times}
\usepackage{color}
\newcommand{\lnub}{\hat{\mathcal{E}}}
\newcommand{\lnlb}{\check{\mathcal{E}}}

\newcommand{\rr}{\mathbbm{R}}
\newcommand{\cc}{\mathbbm{C}}

\newcommand{\1}{\mathbbm{1}}

\newcommand{\rhox}{\rho_\times}
\newcommand{\Gp}{\gamma_+}
\newcommand{\Gm}{\gamma_-}
\newcommand{\Gx}{\gamma_\times}
\newcommand{\dett}{\tilde\det}

\newcommand{\lneg}{\mathcal{E}}
\newcommand{\ek}{\varepsilon_k}
\newcommand{\exk}{\varepsilon^\times_k}

\newcommand{\tr}{{\rm tr}}

\newcommand{\be}{\begin{equation*}}
\newcommand{\ee}{\end{equation*}}
\newcommand{\bea}{\begin{eqnarray*}}
\newcommand{\eea}{\end{eqnarray*}}

\def\>{\rangle}
\def\<{\langle}

\definecolor{zz}{rgb}{1,.6,0}

\newcommand{\qed}{}
 \sbox0{$\begin{matrix}1&2&3\\0&1&1\\0&0&1\end{matrix}$}

\def\qed{\leavevmode\unskip\penalty9999 \hbox{}\nobreak\hfill
     \quad\hbox{\leavevmode  \hbox to.77778em{%
               \hfil\vrule   \vbox to.675em%
               {\hrule width.6em\vfil\hrule}\vrule\hfil}}
     \par\vskip3pt}
     
\begin{document}

\newtheorem{theorem}{Theorem}
\newtheorem{lemma}[theorem]{Lemma}
\newtheorem{corollary}[theorem]{Corollary}
\newtheorem{proposition}[theorem]{Proposition}
\newtheorem{definition}[theorem]{Definition}
\newtheorem{example}[theorem]{Example}
\newtheorem{conjecture}[theorem]{Conjecture}

\title{Entanglement negativity bounds for fermionic Gaussian states}

\author{Jens Eisert} 
\affiliation{Dahlem Center for Complex Quantum Systems, Freie Universit\"at Berlin, 14195 Berlin, Germany}
\author{Viktor Eisler}
\affiliation{Institute of Theoretical and Computational Physics, Graz University of Technology,
Petersgasse 16, 8010 Graz, Austria}
\affiliation{MTA-ELTE Theoretical Physics Research Group,
E\"otv\"os Lor\'and University, P\'azm\'any P\'eter s\'et\'any 1/a, 1117 Budapest, Hungary}
\author{Zolt\'an Zimbor\'as}
\affiliation{Dahlem Center for Complex Quantum Systems, Freie Universit\"at Berlin, 14195 Berlin, Germany}

\date{\today}
\begin{abstract}
The entanglement negativity is a versatile measure of entanglement that has numerous applications in quantum information and in condensed matter theory. It can not only efficiently be computed in the Hilbert space dimension, but for non-interacting bosonic systems, one can compute the negativity efficiently in the number of modes. However, such an efficient computation does not carry over to the fermionic realm, the ultimate reason for this being that the partial transpose of a fermionic Gaussian state is no longer Gaussian. To provide a remedy for this state of affairs, in this work we introduce 
efficiently computable and rigorous upper and lower bounds to the negativity, making use of techniques of semi-definite programming, building upon  the Lagrangian formulation of fermionic linear optics, and exploiting suitable products of Gaussian operators. We discuss examples in quantum many-body theory and hint at applications in the study of topological properties at finite temperature.
\end{abstract}
\maketitle

\section{Introduction}

Entanglement is the distinct feature that makes quantum mechanics fundamentally different from a classical
statistical theory. Undeniably playing a pivotal role in quantum information theory, in notions of key distribution,
quantum computing and simulation, it is increasingly becoming clear that notions of entanglement have the 
potential to add a fresh perspective to the study of systems of condensed matter physics. Notions of 
entanglement entropies and spectra are increasingly used to capture properties of quantum systems with 
many degrees of freedom \cite{Bennett,AreaReview,Haldane}. 
The entanglement entropy based on the von-Neumann entropy plays here presumably
the most important role \cite{Bennett,AreaReview}. 
However, it makes sense as an entanglement measure only for pure states. Hence,
early on, computable measures of entanglement such as the entanglement negativity 
\cite{Volume,PhD,VidalNegativity,PlenioNegativity} have been considered
in the context of the study of quantum many-body systems. In fact, one of the earliest studies
on entanglement properties of ground states of local Hamiltonians considered this entanglement measure
\cite{Harmonic}, which was followed by a series of works on harmonic lattices 
\cite{Area,Area2,Area3,AcinThermal,Janet,MarcovitchRPR08}.

Recent years have seen a revival of interest in studies of entanglement negativity,
and the problem has been attacked using a number of different approaches.
Numerical studies were performed for various spin chains via tensor network calculations
\cite{WMB09,BSB10,NegativityCriticalIsing,RandomSpinChainNegativity},
Monte Carlo simulations where the replica trick comes into play \cite{NegativityMonteCarlo,NegativityReplicaTrick}
or via numerical linked cluster expansion \cite{HastingsNegativity}.
On the analytical side, major developments include the conformal field theory (CFT) approach
\cite{QFTNegativity2,QFTNegativity3} which has also been extended to finite temperature
\cite{NegativityNonEq,CCT15}, non-equilibrium
\cite{NegativityNonEq,NegativityQuantumQuenches,ConformalNegativity2, WCR15}
and off-critical \cite{BCD16} scenarios. For some particular spin chains, there are
even exact results available \cite{WVB10,SKB11,NegativityVBS,OverlapFreeFermionsNegativity}. Studies of negativity
have also been carried out for two-dimensional lattices
\cite{EntanglementNegativityIn2D,NCT16}
with a particular emphasis on topologically ordered phases
\cite{LeeVidal13,Castelnovo13,WCR16,WMR16}

The entanglement negativity -- first proposed in Ref.\ \cite{Volume}, elaborated upon in Ref.\ \cite{EisertPlenioNeg}, 
and proven to be an entanglement monotone in Refs.\ \cite{PhD,VidalNegativity} -- can be computed efficiently in the Hilbert space
dimension for spin systems. For Gaussian bosonic systems, as they occur as ground and thermal states of non-interacting
models, the negativity can even be efficiently computed in the number of modes \cite{VidalNegativity,Harmonic,Continuous,AI07}.
This is possible because the partial transpose \cite{Peres} on which the entanglement negativity is based, reflects partial
time reversal \cite{Simon}, which maps bosonic Gaussian states to Gaussian operators. 
This is in sharp contrast to the situation for fermionic Gaussian systems,
where the partial transpose is, in general, no longer a fermionic Gaussian operator \cite{PartialTransposeGaussian}.
Consequently, there is still no efficiently calculable formula known for the negativity.
This is unfortunate, since Gaussian (or free) fermionic systems are specifically rich.
For example, some well-known models showing features of \emph{topological} properties such as Kitaev's honeycomb lattice model are non-interacting \cite{ToricCode}.  Also one of the most paradigmatic one-dimensional models exhibiting edge states in a topologically non-trivial phase, the Su-Schrieffer-Heeger (SSH) model \cite{SSHmodel} is a 
non-interacting (or quasi-free) fermionic system.

The  lack of a formula for negativity of fermionic Gaussian  states has stimulated a concerted research activity on identifying good bounds \cite{PartialTransposeGaussian, EstimatingNegativityForFreeFermions}.
In this work, we make a fresh attempt at proving tight bounds to the entanglement negativity.
Each bound considered here depend exclusively on the covariance matrix
of the Gaussian state at hand, and thus is efficiently computable in the number of modes.
In particular, the lower bound makes use of a pinching transformation of the covariance matrix,
while the first of two upper bounds requires techniques of semi-definite programming.
The second upper bound was already proposed in a CFT context \cite{EstimatingNegativityForFreeFermions},
which is now elaborated and closed form expressions for arbitrary fermionic Gaussian states are given.
We also test our bounds by estimating the negativity between adjacent segments
in the SSH model and the XX chain, both in the ground state and at finite temperatures.

The paper is structured as follows: In Section \ref{sec:prelim} we introduce the notation used in the rest of this work
and define the negativity,
followed by some basic examples given in Section \ref{sec:basic}. The lower bound is constructed in Section \ref{sec:lb},
whereas Section \ref{sec:ub1} and \ref{sec:ub2} deal with two different upper bounds, based on semi-definite programming
and products of Gaussian operators, respectively. Numerical checks of the bounds are presented in Section \ref{sec:num},
followed by our concluding remarks in Section \ref{sec:outlook}.

\section{Preliminaries\label{sec:prelim}}
\subsection{Fermionic quantum systems}

Throughout this work, we consider quantum systems consisting of a set of fermionic modes; the annihilation and creation operators $\{f_1^{\phantom\dagger}, f_1^\dagger, \ldots f_k^{\phantom\dagger}, f_k^\dagger\}$ associated with the modes generates the CAR algebra, i.e., the algebra of operators respecting the canonical anti-commutation relations. In many context it is convenient to refer rather to Majorana fermions than to the original ones, by defining
\begin{equation}
m_{2j-1}= f_j^{\dagger} + f_j^{\phantom\dagger} \, , \; \;  m_{2j}=i( f_j^{\dagger} - f_j^{\phantom\dagger})
\end{equation}
for $j=1,\dots, k$. Given a state $\rho$, the second moments of the Majorana fermions can be collected in the covariance matrix
$\gamma\in \rr^{2k\times 2k}$, with entries
\begin{equation}
	\gamma_{j,l} = \frac{i}{2}{\rm tr}(\rho[m_j, m_l]).
\end{equation}
 It is easy to see that this matrix satisfies
\begin{equation}
	\gamma=-\gamma^T,\, i \gamma\leq \1.
\end{equation}
We will denote the set of such covariance matrices of $k$ modes as $C_k\subset \rr^{2k\times 2k}$. 

A fermionic Gaussian state $\rho$ is completely defined by its covariance matrix, as 
one can express the expectation value of any Majorana monomial through the Wick expansion
\begin{equation}
\tr(\rho \, m_{j_1} m_{j_2} \ldots m_{j_{2p}}){=} (-i)^{p} \sum_{\pi} {\rm{ sgn}}(\pi) \prod_{l=1}^{p} \gamma_{j_{\pi(2l-1)}, j_{\pi(2l)}},
\end{equation}
where the indices of the Majorana operators are different and the sum runs over all pairings $\pi$ (with ${\rm{ sgn}}(\pi)$ denoting the sign of the pairing).

Considering a Gaussian (or quasi-free, as it is also called) 
unitary 
\begin{equation}
V{=}\,e^{-\frac{i}{4}\sum_{j,l}K_{j,l} m_j m_l} 
\end{equation}
(where $K\in \rr^{2k\times 2k}$ with $K=-K^T$) and a Gaussian state $\rho$, the evolved state $\rho'=V\rho \, V^\dagger$ remains Gaussian. On the 
level of the covariance matrices, this mapping can be represented by the transformation
\begin{equation} \label{eq:mode-trans}
\gamma \mapsto O_K \gamma \, O^T_K,
\end{equation}
where $O_K=e^{-iK} \in SO(2k)$. 
 In this context, a commonly used tool is that a covariance matrix can be brought to a normal form by means of such a special orthogonal mode transformation $\widetilde{O}$,
\begin{equation}
	\widetilde{O}\gamma \widetilde{O}^T = \bigoplus_{j=1}^k x_j\left[
	\begin{array}{cc}
	0 & -1\\
	1 & 0
	\end{array}
	\right], 
\end{equation}
with $x_j\in [-1,1]$ corresponding to the presence or absence of a fermion in the normal mode decomposition.

A Gaussian state is called particle-number conserving if it commutes with the particle-number operator $\sum_{j=1}^k f^{\dagger}_j f^{\phantom{\dagger}}_j$. In this case the expectation values of the pairing operators vanish, i.e. $\langle f^{\phantom{\dagger}}_j f^{\phantom{\dagger}}_l\rangle= \langle f^{\dagger}_j f^{\dagger}_l\rangle=0$. Thus, the $2k \times 2k$ covariance matrix $\gamma$ can be completely recovered from 
 the $k \times k$ correlation matrix $C_{j,l}=\langle f^\dag_j f^{\phantom{\dagger}}_l\rangle$. 
 Moreover, such a state remains particle-number conserving and Gaussian under a mode-transformations of the form $e^{-i\sum_{j,l}R_{j,l} f^{\dagger}_j f^{\phantom{\dagger}}_l}$ (where $R$ is a Hermitian matrix), and the corresponding map on the correlation matrix level, analogue of Eq.~\eqref{eq:mode-trans},  is given by
\begin{equation}
C \mapsto U_R C \, U^{\dagger}_R,
\end{equation}
where  $U^{\phantom{\dagger}}_R=e^{-iR} \in U(k)$.

\subsection{Partial transpose and negativity}

Let us now turn to the definition of entanglement negativity. Consider a bipartite fermionic system composed of two subsystems $A$ and $B$ corresponding to Majorana modes $\{m_1, \ldots m_{2n}\}$ and $\{m_{2n+1}, \ldots m_{2k}\}$, respectively. Following the literature, we will refer to such a set-up as a bipartite system of $n \times (k-n)$ modes.
Given a bipartite fermionic state $\rho$, the entanglement negativity is defined as
\begin{equation}
	\mathcal{N} = \frac{1}{2}(\|\rho^{T_B}\|_1-1),
\end{equation}
where  $\|.\|_1$ is the trace norm and the superscript $T_B$ denotes partial transposition
with respect to subsystem $B$.
The logarithmic negativity as a derived quantity is 
\begin{equation}
	\mathcal{E} = \ln \|\rho^{T_B}\|_1.
\end{equation}
Both quantities have their significance, and the latter is an entanglement monotone despite not being convex
\cite{PlenioNegativity}, as well
as an upper bound to the distillable entanglement. Since at the heart of the problem under consideration here is the
assessment of 
$ \|\rho^{T_B}\|_1$, a bound to the latter gives immediately a bound to both the negativity and the
logarithmic negativity.

To proceed, we first need to represent the action of the partial transposition
on the density operator.
Using the notations $m_{j}^{0}=\1$ and $m_{j}^{1}=m_j$,  a fermionic state  can be written as 
\begin{equation} 
\rho= \sum_{\tau}
w_{\tau} 
m_{1}^{\tau_1} \ldots m_{2k}^{\tau_{2k}} \, ,
\end{equation}
where the summation runs over all bit-strings
${\tau}=(\tau_{1}, \ldots , \tau_{2k})\in \{0,1\}^{\times 2k}$ 
of length $2k$.\footnote{Note that a physical fermionic state must also commute 
with the parity operator $P=\prod_{j=1}^{2k} m_j$, i.e., one has $w_{\tau}$=0 when 
$\sum_{j=1}^{2k} \tau_{j}$ is odd.}
The partial transpose of $\rho$ with respect to  to subsystem $B$ is the transformation that leaves the $A$ Majorana modes invariant and acts as a transposition $\mathcal{R}$
on the operators built up from modes of $B$, i.e.
\begin{equation}
\rho^{T_B}=
\sum_{\tau} 
w_{\tau} \, 
m_{1}^{\tau_1}  \ldots m_{2n}^{\tau_{2n}}
\mathcal{R}(m_{2n+1}^{\tau_{2n+1}} \ldots m_{2k}^{\tau_{2k}}) \, .
\end{equation}

As shown in Ref. \cite{PartialTransposeGaussian}, the action of $\mathcal{R}$
in a suitable basis can be written as
\begin{equation}
\mathcal{R}(m_{2n+1}^{\tau_{2n+1}} \ldots m_{2k}^{\tau_{2k}})=
(-1)^{f(\tau)}m_{2n+1}^{\tau_{2n+1}} \ldots m_{2k}^{\tau_{2k}}, 
\end{equation}
where
\begin{equation}
f(\tau)=
\begin{cases}
0 & \mbox{if $\sum_{j=2n+1}^{2k}\tau_j \; \mathrm{mod} \; 4 \in \{0,1\}$},\\[1mm]
1 & \mbox{if $\sum_{j=2n+1}^{2k}\tau_j \; \mathrm{mod} \; 4 \in \{2,3\}$}. 
\end{cases}
\end{equation}
As a main consequence one finds that, in sharp contrast to their bosonic counterparts,
the partial transpose operation for fermionic Gaussian states does not
preserve Gaussianity. Nonetheless, in a suitable basis the partial transpose
can still be decomposed as the linear combination of two Gaussian
operators \cite{PartialTransposeGaussian}.

%
%

\section{Basic instances\label{sec:basic}}

When discussing the negativity of Gaussian states, the situation of two fermionic modes is particularly instructive and 
and will be made use of later extensively. 
We hence treat this case in significant detail.

Any two-mode covariance matrix can be brought  into the form
\begin{equation} \label{eq:norma_form}
	\gamma=\left[\begin{array}{cccc}
	0 & a & 0 & -b\\
	-a & 0 & -c & 0\\
	0 & c & 0 & d\\
	b & 0 & -d & 0\\
	\end{array}
	\right],
\end{equation}
referred to as normal form,
upon conjugating with $O_A\oplus O_B$, with $O_A,O_B\in SO(2)$, reflecting a local mode transformation in subsystems labelled $A$ and $B$. Such local mode transformation do not change the entanglement content of the state, and for a Gaussian state with a covariance matrix given by Eq.~\eqref{eq:norma_form}, one can easily compute the negativity. 
This is possible because one can identify the two-qubit system that reflects this Gaussian state, by virtue of the Jordan-Wigner transformation. This two-qubit quantum state is given by the following expression:

\begin{lemma}[Negativity of two modes] Let $\gamma\in C_2$ be a covariance matrix in normal form. The
negativity of the quantum state is that of the state 
\begin{equation}
	\rho=\frac{\1}{4}+\frac{1}{4} \left[
	\begin{array}{cccc}
	M_{1,1} & 0 & 0 & M_{1,4}\\
	0 & M_{2,2} & M_{2,3} & 0\\
	0 & M_{3,2} & M_{3,3}& 0\\
	M_{4,1} & 0 & 0 & M_{4,4}\\
	\end{array}
	\right],
\end{equation}
of two qubits, where 
\begin{align}
	M_{1,1} {=} {-}(a{+}d){+}(ad{+}bc), \, \, &M_{2,2} {=} (a{-}d){-}(ad{+}bc),\\
	M_{3,3} {=} {-}(a{-}d){-}(ad{+}bc), \,  \, &M_{4,4} {=} (a{+}d){+}(ad{+}bc),\\
	M_{1,4}{=}M_{4,1}= b+c,    \; \; \; \; \; \; \; \;  \; &M_{2,3}{=}M_{3,2}=b-c \, .
\end{align}
\end{lemma}
Hence, the negativity of this state can be computed in closed form solving a simple quadratic problem.
It is given by 
\begin{equation} \label{eq:two-mode}
\mathcal{N} = \frac{1}{2}(\| \rho^{T_B}\|_1-1)=\frac{1}{2}( h(\gamma)-1),
\end{equation}
where we defined the function 
\begin{align}
h(\gamma){=} \frac{1}{2}+\frac{1}{2}\max\{1&, \sqrt{(a{+}d)^2 {+}(b{-}c)^2 }{-}(ad{+}bc) \nonumber \\
&,  \sqrt{(a{-}d)^2 {+}(b{+}c)^2 }{+}(ad{+}bc) \}. \label{eq:two-mode2}
\end{align}

\subsection{Fermionic Gaussian  pure-state entanglement}

A Gaussian state is pure iff $\gamma^2=-\1$. In a $1 \times 1$ set-up this implies that by conjugating $\gamma$ with 
a local mode transformation $O_A \oplus O_B$ (where $O_A,O_B\in SO(2)$), one can bring it into a Bardeen-Cooper-Schrieffer (BCS) form
\begin{equation}
	\gamma(a) =\left[\begin{array}{cccc}
	0 & a & 0 & -b\\
	-a & 0 & -b & 0\\
	0 &b & 0 & a\\
	b & 0 & -a & 0\\
	\end{array}
	\right],
\end{equation}
with $b:=(1-a^2)^{1/2}$. Thus, the state depends on a single parameter $a\in [-1,1]$, and  its negativity is given by
\begin{equation}\label{eq:preg}
	\mathcal{N} =\frac{1}{2}( \| \rho^{T_B}\|_1-1)=\frac{1}{2}(g(a) -1),
\end{equation}
where we defined 
\begin{equation}\label{eq:g}
g(a)= 1+ \sqrt{1-a^2}.
\end{equation}
For a multi-mode fermionic Gaussian  pure state, this gives rise to an explicit 
simple expression for the negativity, which we state in the following lemma.

\begin{lemma}[Pure fermionic Gaussian  states] \label{lem:pure} The negativity of a pure fermionic Gaussian 
state of $n\times n$ modes is
\begin{equation} \label{eq:pure_neg}
 \mathcal{N} = \frac{1}{2}\left(\prod_{j=1}^n g(a_j)-1\right),
\end{equation}
where $\{\pm i a_j\}$ is the spectrum of $\gamma_A$.
\end{lemma}

\begin{proof}
It is known that for any covariance matrix satisfying $\gamma^2=-\1$ can be brought into a multi-mode BCS
form \cite{BR04}
\begin{align}
	&(O_A\oplus O_B) \gamma(O_A\oplus O_B)^T  =  \widetilde{\oplus}_{j=1}^n \gamma(a_j)= \nonumber \\[2mm]
	& \left[
\begin{array}{c@{}|c@{}}
\bigoplus_{j=1}^{n} \left[\begin{array}{cc}
         0 &  a_{j} \\
          - a_{j} & 0 \\
  \end{array}\right] \; &  
  \; \bigoplus_{j=1}^{n} \left[\begin{array}{cc}
         0 & -b_{j} \\
         -b_{j} & 0 \\
  \end{array}\right] \\[5mm]
  \hline\\[-2mm]
  \bigoplus_{j=1}^{n} \left[\begin{array}{cc}
         0 & b_{j} \\
         b_{j} & 0 \\
  \end{array}\right] \; &  
  \; \bigoplus_{j=1}^{n} \left[\begin{array}{cc}
         0 &  a_{j} \\
          - a_{j} & 0 \\
  \end{array}\right]\\
\end{array}\right],
\end{align}	
where $\widetilde{\oplus}$ denotes a direct sum giving the above type of block structure, $O_A,O_B\in SO(2n)$, $\{ \pm i a_j \}$ is the spectrum of  $\gamma_A$,  and $a_j^2+b_j^2=1$. In other words, one can decouple the modes in $A$ and $B$ such that there is entanglement only between the corresponding pairs. Thus, we can write (after rearranging the modes) the state as a product of these pairwise entangled $1 \times 1$-mode states.
Using the multiplicativity of the trace norm and
and 
the negativity formulas Eqs.~\eqref{eq:preg} and \eqref{eq:g} for each of the decoupled $1 \times 1$ mode pairs, we arrive immediately at Eq.~\eqref{eq:pure_neg}.\qed
\end{proof}
Let us also note that as for general pure states $\rho$,
\begin{equation}
	\|\rho^{T_B}\|_1 = {\rm tr}(\rho_A^{1/2})^2
\end{equation}
holds true, the negativity could anyway efficiently be computed via standard formulas for R\'enyi entropies of Gaussian states \cite{LRV04, JK04}, yielding the same formula as Eq.~\eqref{eq:pure_neg}.

For the sake of completeness, we mention that one can generalize the above results for any Gaussian state 
that can be brought by a local mode transformation into a state with the following type of covariance matrix:
\begin{align}
	& \left( 
	\begin{array}{c@{}|c@{}}
\bigoplus_{j=1}^{n} \left[\begin{array}{cc}
         0 &  a_{j} \\
          - a_{j} & 0 \\
  \end{array}\right] \; &  
  \; \bigoplus_{j=1}^{n} \left[\begin{array}{cc}
         0 & -b_{j} \\
         -c_{j} & 0 \\
  \end{array}\right] \\[5mm]
  \hline\\[-2mm]
  \bigoplus_{j=1}^{n} \left[\begin{array}{cc}
         0 & c_{j} \\
         b_{j} & 0 \\
  \end{array}\right] \; &  
  \; \bigoplus_{j=1}^{n} \left[\begin{array}{cc}
         0 &  d_{j} \\
          - d_{j} & 0 \\
  \end{array}\right]\\
\end{array}\right).
\end{align}
For states with such properties (e.g., for the isotropic states \cite{BR04}), the negativity can be calculated using the general two-mode formula Eq.~\eqref{eq:two-mode}, the final result being
\begin{equation}
\mathcal{N}=\frac{1}{2}\left( \prod_{j=1}^n h(\gamma_j) -1\right),
\end{equation}
where $h(\gamma_j)$ is defined as in Eq.~\eqref{eq:two-mode2} with the corresponding parameters $a_j, b_j,c_j,d_j$.

\section{Lower bound \label{sec:lb}}

We now turn to presenting bounds to the entanglement negativity for arbitrary fermionic Gaussian states. We first discuss a lower bound, before proceeding to the more sophisticated upper bounds. The lower bound will be derived from a pinching transformation using the expression of two-mode negativity reviewed in the previous section.

\subsection{Lower bound from pinching}

Using the pinching transformation, one can decouple the system into independent $1\times 1$ modes, and use
for each of these system the previously obtained expression for the negativity for the $1\times 1$ case.
In the obtained expression $\pi_j$ denotes the $4\times4$-submatrix associated with the 
respective $j$-th $1\times 1$ subsystems.

\begin{theorem}[Lower bound] \label{lem:lb} An efficiently  computable lower bound of the negativity of a fermionic Gaussian state $\rho$ of $n\times n$ modes 
with covariance matrix $\gamma$ is for every $O_A,O_B\in SO(2n)$ provided by  
\begin{equation}
	\mathcal{N}(\rho)\geq \frac{1}{2}\left(\prod_{j=1}^n h(\pi_j(O_A\oplus O_B \gamma O_A^T\oplus O_B^T))-1\right).
\end{equation}
\end{theorem}
\begin{proof}
In particular, $O_A=O_B=\1$ is a legitimate choice in this bound.
The above statement follows from  the fact
that making use of random phases, one can group twirl the conjugate covariance matrix 
$\Gamma:= O_A\oplus O_B \gamma O_A^T\oplus O_B^T$
into
\begin{equation}
	\Gamma':= \bigoplus_{j=1}^n \pi_j(\Gamma),
\end{equation}
for which the negativity can be readily computed as stated above. The group twirl 
amounts to a map
\begin{equation}
	\Gamma\mapsto \Gamma'= \frac{1}{n}\sum_{j=1}^n O_j \Gamma O_j^T
\end{equation}
on the level of covariance matrices, where
\begin{equation}
	O_j:= {\rm diag}(H_j)\otimes \1_4.
\end{equation}
In this expression $H_j$, $j=1,\dots, n$, is the $j$-th row of a real Hadamard matrix 
\begin{equation}
	H\in \{-1,1\}^{n\times n}\in O(n),
\end{equation}
so an orthogonal matrix the entries of which are $\pm1$. This is  to show that blocks of four Majorana operators
each are equipped with signs, so that the resulting covariance matrix has the desired pinched form. The above
group twirl can be performed with local operations and classical communication, hence it provides a 
lower bound, making use of the fact that the negativity is an entanglement monotone.\qed
\end{proof}

By choosing appropriate $O_A$ and $O_B$ (e.g., through an optimization procedure), one may obtain useful bounds for the entanglement negativity. The case of particle-number conserving Gaussian states is especially tractable. 

\subsection{The particle number conserving case}

As discussed in Section~\ref{sec:prelim}, when treating  particle-number conserving Gaussian states, instead of the covariance matrix $\gamma$, we can work with the correlation matrix $C_{j,l}=\langle f^\dag_j f^{\phantom{\dagger}}_l\rangle$. When $C$ is real, one has the very simple relation 
\begin{equation}
\gamma_{2j-1, 2l}=-\gamma_{2l, 2j-1} =2C_{j,l} - \delta_{j,l} \, ,
\end{equation}
with all the other entries of $\gamma$ being zero.

 
 Considering an $n \times n$ set-up, we can divide the total correlation matrix of a state $\rho_{A \cup B}$  with respect to  the two subsystems:
\begin{equation}
C =\left[
\begin{array}{c@{}|c@{}}
\; C_{A,A} \; & \; C_{A,B} \; \\[2mm]
\hline \\[-4mm]
  C_{B,A}& C_{B,B}
\end{array}
\right],
\end{equation}
where $C_{A,A}$ and $C_{B,B}$ are Hermitian, and $C^{\dagger}_{A,B}=C_{B,A}$.
Let us choose the particle-number conserving local mode transformation $U_A \oplus U_B$ such 
that $U_A C_{A,B} U^{\dagger}_{B}$ is a positive diagonal matrix, i.e., $U_A$ and $U^{\dagger}_B$ provide the singular value decomposition of $C_{A,B}$.  Applying now a pinching transformation on the mode-rotated state, we obtain a Gaussian state $\rho'_{A \cup B}$ for which
\begin{equation}
2C'-\1=\left[
\begin{array}{c@{}|c@{}}
\begin{array}{ccc}
	a_1 &  &  \\
	 & \ddots &  \\
	 &  & a_n  \\
	\end{array} & \begin{array}{ccc}
	c_1 &  &  \\
	 & \ddots &  \\
	 &  & c_n  \\
	\end{array} \\
\hline 
  \begin{array}{ccc}
	c_1 &  &  \\
	 & \ddots &  \\
	 &  & c_n  \\
	\end{array}& \begin{array}{ccc}
	d_1 &  &  \\
	 & \ddots &  \\
	 &  & d_n  \\
	\end{array}
\end{array}
\right],
\end{equation}
where the non-diagonal elements of the block matrices are all zero, and $a_j, d_j$ and $c_j$ denote 
the diagonal elements of the matrices $(2U_A C_{A,A} U_A^{\dagger}{-}\1)$,   $(2U_B C_{B,B} U_B^{\dagger}{-}\1)$, and $2U_A C_{A,B} U^{\dagger}_{B}$, respectively. Now, using Lemma~\ref{lem:lb}, we obtain the following lower bound for the negativity for the original Gaussian state $\rho_{A \cup B}$ 
\begin{equation}
\mathcal{N} (\rho_{A \cup B}) \ge \frac{1}{2}\left( \prod_{j=1}^n h(\gamma_j) -1\right),
\end{equation}
where
\begin{equation} \label{eq:norma_form}
	\gamma=\left[\begin{array}{cccc}
	0 & a_j & 0 & c_j\\
	-a_j & 0 & -c_j & 0\\
	0 & c_j & 0 & d_j\\
	-c_j & 0 & -d_j & 0\\
	\end{array}
	\right].
\end{equation}
In Section~\ref{sec:num} we will use this procedure to numerical calculate lower bounds for the negativity in the ground and thermal states of various many-body systems.

\section{Upper bound via convex optimization\label{sec:ub1}}

Good upper bounds for the entanglement negativity are significantly harder to come by, and they constitute the main result of this paper. This section presents the first of our two novel strategies to arrive at upper bounds. It is rooted in ideas of convex optimization and the structure theorem of Gaussian maps obtained from the Lagrangian formulation
of fermionic linear optics. 

\subsection{Upper bounds via convex optimization}

 The basic idea of this bound is to make use of the fact that the negativity is an entanglement monotone, 
 meaning that by means of local transformation, 
the entanglement content cannot increase on average. In this way, an upper bound can be identified
once one is in the position to identify those Gaussian root states from which the desired
state can be prepared. As it turns out, this gives rise to a problem that can be tackled with
the machinery of convex optimization. The bound as such will require some preparation, however. We start by
stating how fermionic Gaussian maps, so not necessarily trace-preserving completely
positive maps that send Gaussian states to Gaussian states, act on the level of covariance matrices.

\begin{theorem}[Structure of fermionic Gaussian maps \cite{LagrangianBravyi}]\label{lb} An 
arbitrary fermionic Gaussian operation acts on covariance matrices $\gamma\in C_m$ as
\begin{equation}\label{Law}
	\gamma\mapsto B(\gamma^{-1}+ D)^{-1} B^T +A,
\end{equation}
where 
\begin{equation}\label{cm}
	\Gamma := \left[\begin{array}{cc}
	A & B\\
	-B^T & D
	\end{array}
	\right]\in C_{2m}
\end{equation}
is a fermionic covariance matrix on a doubled mode space.
\end{theorem}

We now turn to an observation that is helpful in this context: That all outcomes in a selective fermionic Gaussian map
are related with each other upon conjugating the input with a diagonal matrix $P$ from $P_m$, with 
\begin{equation}
	P_m:=\left\{P = \bigoplus_{j=1}^m x_i\1_2, \, x_i\in \{-1,1\}\right\}.
\end{equation}
This feature 
mirrors a similar property in the Gaussian bosonic setting, where with an appropriate shift in phase space conditioned on the
measurement outcome, an arbitrary Gaussian map can be made trace-preserving \cite{GaussianNoGo,GiedkeGaussian}.

\begin{lemma}[Selective fermionic Gaussian operations] For any selective fermionic Gaussian operation, one 
outcome being described by a map (\ref{Law}), the other measurement outcomes are reflected by covariance matrices of the form
\begin{equation}
	\gamma\mapsto B ( P\gamma^{-1}P  + D )^{-1}  B^T +A,
\end{equation}
where $P\in P_m$.
\end{lemma}

\begin{proof}
This means that all outcomes of a selective fermionic Gaussian map 
are on the level of covariance matrices reflected by the same transformation,
 upon conjugating the input by a matrix $P\in P_m$. 
This can be seen by acknowledging the fact that any post-selected
fermionic completely positive map can be written as a concatenation of a 
fermionic Gaussian channel, acting as 
\begin{equation}
	\gamma\mapsto X\gamma X^T + Y,
\end{equation}
with 
$Y=-Y^T$, $X X^T\leq \1$ and $i Y\leq \1-X X^T$, in addition to dilations
\begin{equation}
	\gamma\mapsto O(\gamma\oplus \gamma')O^T,
\end{equation}
with $\gamma'\in C_{k}$, $O\in SO(2(m+k))$, followed by a fermion number measurement
on the additional $k$ modes. This follows from Ref.\ \cite{LagrangianBravyi}, mirroring the situation
for bosonic post-selected Gaussian completely positive maps \cite{GaussianNoGo,GiedkeGaussian}.
For different outcomes of that fermionic measurement, the above map is being replaced by
 $(\1\oplus P)\Gamma (\1\oplus P)$, $P\in P_m$.  
This means that for different measurement outcomes,
the map in Eq.\ (\ref{Law}) is being replaced
by 
\begin{equation}
	\gamma\mapsto BP(\gamma^{-1}+ PD P)^{-1} P B^T +A.
\end{equation}
This is identical with 
\begin{equation}
	\gamma\mapsto B ( P\gamma^{-1}P  + D )^{-1}  B^T +A.
\end{equation}
\qed
\end{proof}

This structure can be uplifted to the level of local fermionic Gaussian operations,
which seems helpful in its own right.

\begin{lemma}[Local fermionic Gaussian operations] 
Each outcome of a selective local fermionic Gaussian  operation on an $n\times n$ system gives rise to a
covariance matrix of the form
\begin{equation}
	\gamma\mapsto B(P \gamma^{-1}P + D)^{-1} B^T +A,
\end{equation}
where $A,B,D\in \rr^{2n\times 2n}$ are submatrices of a covariance matrix as in Eq.\ (\ref{cm}) with $m=2n$ and
$A=A\oplus B$, 
$B=B_1\oplus B_2$, 
$D=D_1\oplus D_2$, and $P\in P_{2n}$.
\end{lemma}

\begin{proof}
This structure follows immediately from the above characterisation of fermionic Gaussian maps.\qed
\end{proof}

\subsection{Upper bound }

We are now in the position to develop the idea for the upper bound. The basic idea is that we would like to identify a simple $\xi\in C_{2n}$, constituted of blocks of $4\times 4$-matrices that reflect entangled pairs of fermionic modes,  
such that 
\begin{equation}
	\gamma = B(\xi^{-1}+ D)^{-1} B^T +A,
\end{equation}
reflecting a local fermionic Gaussian operation.  Using the monotonicity of the negativity, this gives rise to a
tight upper bound.

\begin{theorem}[Upper bound for the negativity] An efficiently  computable upper bound of the negativity of a fermionic Gaussian state $\rho$ of $n\times n$ modes  with covariance matrix $\gamma$ is given by the solution of the semi-definite problem 
\begin{equation}
 \min  v:= \sum_{j=1}^n v_j 
\end{equation}
subject to 
\begin{align}
	& v_j = | {\rm tr}(G \eta_j) |,\\
	& i \left[ 
	\begin{array}{cc}
	\gamma -  A & B \\
	- B^T &  \eta  + D
	\end{array}
	\right] \geq 0,\label{cf}\\
	& A = A\oplus B, \\
	& B = B_1\oplus B_2,\\ 
	& D = D_1\oplus D_2,\\
	& \eta = \bigoplus_{j=1}^n \eta_j,\\
	&\eta_j = -\left[
\begin{array}{cccc}
0 & \alpha_j & 0 &  -\beta_j\\
-\alpha_j & 0 & -\beta_j & 0\\
0 & \beta_j & 0 & \alpha_j\\
\beta_j &0& -\alpha_j & 0\\
\end{array}\right],\\
	& i\eta  \geq  \1,\\
	&\eta =  -\eta^T,\\
	&i \left[\begin{array}{cc}
	A & B\\
	-B^T & D
	\end{array}
	\right] \leq  \1,\label{ce}
\end{align}
where
\begin{equation}
	G:= \bigoplus_{j=1}^{2}
	\left[\begin{array}{cc}
	0 & 1 \\
	-1 & 0\\
	\end{array}\right],
	\end{equation}
	as 
	\begin{equation}
		\mathcal{N} \leq \sum_{j=1}^n \left(
		- \frac{1}{2}+ \frac{1}{8}(16-v_j^2)^{1/2}\right).
	\end{equation}
\end{theorem}

\begin{proof}
The logic of this argument is that  the entanglement content of the Gaussian state
described by the covariance matrix $\xi$ must be larger than that of $\gamma$, invoking the fact that the negativity is an entanglement monotone \cite{PhD,VidalNegativity}. We can build upon the above characterization of fermionic Gaussian maps.
What is more, each other outcome is related to the above upon conjugating $\xi$ with a $P$ of the above form.

We start from a $\xi\in C_{2n}$, constituted of blocks $\xi_j$ of $4\times 4$ for $j=1,\dots, n$. 
These covariance matrices are taken to be of the form
\begin{equation}
\xi_j = \left[
\begin{array}{cccc}
	0 & a_j & 0 &-b_j\\
	-a_j & 0 & -b_j & 0\\
	0 &b_j &0 & a_j\\
	b_j& 0  &-a_j & 0\\
\end{array}
\right],
\end{equation}
with $a_j^2+b_j^2\leq 1$. 
If a local fermionic Gaussian operation can be found, then for some suitable 
$A=A\oplus B$, $B=B_1\oplus B_2$,
$D=D_1\oplus D_2$, and a $P\in P_{2n}$ one has 
\begin{equation}
	i\gamma = i B(P \xi^{-1}P + D)^{-1} B^T +i A,
\end{equation}
which can be relaxed into an inequality
\begin{equation}\label{BS}
	i\gamma \geq i B(P \xi^{-1}P + D)^{-1} B^T + i A,
\end{equation}
The inverse is hard to handle in this expression, which is why we continue to incorporate the inverse
directly into the convex program.
Defining $\eta:= \xi^{-1}$, the constraint $i\xi\leq \1$ becomes 
	\begin{equation}
	i\eta\geq \1.
\end{equation}
We can now make use of a Schur complement \cite{Bhatia} to relate (\ref{BS}) to a positive semi-definite constraint:
The validity of 
\begin{equation}
	i \left[ 
	\begin{array}{cc}
	\gamma-A & B\\
	-B^T & P \eta P +D
	\end{array}
	\right]\geq 0
\end{equation}
also implies the validity of (\ref{BS}). At this point, the 
relaxed constraints become
\begin{align}
	&i \left[ 
	\begin{array}{cc}
	\gamma -  A & B P\\
	-P B^T &  \eta  +P DP
	\end{array}
	\right]\geq 0,\label{cf}\\
	&A=A\oplus B, \\
	&B=B_1\oplus B_2,\\ 
	&D=D_1\oplus D_2,\\
	&\eta= \bigoplus_{j=1}^n \eta_j,\\
	&i\eta  \geq  \1,\\
	&\eta =  -\eta^T,\\
	&i \left[\begin{array}{cc}
	A & B\\
	-B^T & D
	\end{array}
	\right] \leq  \1.\label{ce}
\end{align}
We can impose the explicit form
\begin{equation}
\eta_j = -\left[
\begin{array}{cccc}
0 & \alpha_j & 0 &  -\beta_j\\
-\alpha_j & 0 & -\beta_j & 0\\
0 & \beta_j & 0 & \alpha_j\\
\beta_j &0& -\alpha_j & 0\\
\end{array}\right],
\end{equation}
of the inverses with $\alpha_j,\beta_j\in \rr$.
The negativity cannot be directly cast into a convex problem. However, we can make use of yet another convex relaxation,
in order to arrive at an efficiently computable upper bound. As this involves some steps, this is separately
laid out in Lemma \ref{UB}. 
The final statement follows from the fact that the $P\in P_{2n}$ has no significance in the bound, and hence we can optimise for $P\in \1_{4n}$. This ends the argument.
\qed
\end{proof}

As an example, let us discuss what this negativity upper bound gives for two-mode systems.

\begin{lemma}[Upper bound to two-mode entanglement]\label{UB} 
Let  
\begin{equation}
\eta = 
-\left[
\begin{array}{cccc}
0 & \alpha & 0 &  -\beta\\
-\alpha & 0 & -\beta & 0\\
0 & \beta & 0 & \alpha\\
\beta &0& -\alpha & 0\\
\end{array}\right],
\end{equation}
with $i\eta\geq \1$. Then the inverse $\eta^{-1}\in C_2$ is a fermionic covariance matrix 
and the negativity of the fermionic Gaussian state is upper bounded by 
\begin{equation}
\mathcal{N}\leq- \frac{1}{2}+ \frac{1}{8}(16-{\rm tr}(G \eta^{-1})^2)^{1/2},
\end{equation}
where
\begin{equation}
	G= \bigoplus_{j=1}^2
	\left[\begin{array}{cc}
	0 & 1 \\
	-1 & 0\\
	\end{array}\right].
	\end{equation}
\end{lemma}

\begin{proof}
The inverse of $\eta$ is easily identified to be
\begin{equation}
	\eta^{-1} = (\alpha^2+ \beta^2)^{-1}
	\left[
	\begin{array}{cccc}
	0 & \alpha & 0 &  -\beta\\
	-\alpha & 0 & -\beta & 0\\
	0 & \beta & 0 & \alpha\\
	\beta &0& -\alpha & 0\\
	\end{array}\right].
\end{equation}
If $i\eta\geq \1$, then $\eta^{-1}$ is a covariance matrix in $C_2$.
It follows immediately that
\begin{equation}	
	 |{\rm tr}(G \eta)| = 4|\alpha|.
\end{equation}	
The value of $|\alpha|$ clearly gives rise to an upper bound for 
\begin{equation}
	|\alpha| (\alpha^2+ \beta^2)^{-1}.
\end{equation}
This in turn gives rise to a bound to the negativity, acknowledging again the connection of covariance matrices
in $C_2$ and spin states of two spins or qubits in $(\cc^2)^{\otimes 2}$. For the first qubit, a value of
$|\alpha|$ implies that in 
\begin{equation}
\lambda|0\rangle\langle 0| + (1-\lambda)|1\rangle\langle 1|.
\end{equation}
one has that $(|\alpha|+1)/2=\lambda$.
This gives an upper bound to the negativity, making use of its convexity. Asserting that
\begin{equation}
\rho=\left[
\begin{array}{cccc}
	 E^2 & 0 &0 &  EF\\
	0 & 0 & 0 & 0\\
	0 & 0 & 0 & 0\\
	EF &0& & F^2\\
	\end{array}\right],
\end{equation}
with $E,F\in \rr$ satisfying $E^2+F^2=1$,
one finds that
\begin{equation}
	\mathcal{N}\leq \lambda^{1/2} (1-\lambda)^{1/2}-\frac{1}{2}.
\end{equation}
From this the above statement follows. \qed
\end{proof}

\section{Upper bound from products of Gaussian operators\label{sec:ub2}}

We now turn to a second upper bound to the entanglement negativity, which complements the previous one and
that serves a quite different aim. It can again efficiently computed and allows for bounding the 
entanglement negativity in large systems. We now consider a system of $n$ modes,
where now the modes are separated into subsets $A$ and $B$. We no longer require $A$ and $B$ to have the
same cardinality, but can also allow for arbitrary cuts into a system $A$ and its complement.

For any such division, we can define the operators $O_{\pm}$ as the Gaussian operators -- which do not necessarily
reflect quantum states -- that have the fermionic covariance matrix 
\begin{equation}
\gamma_{\pm} = T^{\pm}_{B} \, \gamma \, T^{\pm}_{B},
\label{gpm}
\end{equation}
where
\begin{equation}
T^{\pm}_{B} =\bigoplus_{j\in A}
\1_2 \bigoplus_{j\in B}
({\pm i}) \1_2 .
\end{equation}
In other words, $O_{\pm}$ are defined as the Gaussian operators satisfying
\begin{equation}
\frac{i}{2}{\rm tr}\, (O_{\pm} \, [m_j, m_k])= (\gamma_{\pm})_{j,k}.
\end{equation}
Using this definition, the partial transpose of a Gaussian state can be written in the form
\cite{PartialTransposeGaussian}
\begin{equation}
\rho^{T_B}= \frac{1-i}{2} O_+ + \frac{1+i}{2} O_- \, .
\label{rhotb}
\end{equation}

The main difficulty in evaluating the trace norm of the partial transpose is that its
constituent Gaussian operators $O_+$ and $O_-$ do not commute in
general, and thus one has no direct access to the spectrum of $\rho^{T_B}$.
Nevertheless, the simple form of Eq.\ \eqref{rhotb} allows one to apply a
triangle inequality to bound the trace norm as \cite{EstimatingNegativityForFreeFermions}
\begin{equation}
\| \rho^{T_B} \|_1 \le \left\| \frac{1-i}{2} O_+ \right\|_1 + \left\| \frac{1+i}{2} O_- \right\|_1 =
\sqrt{2} \| O_+ \|_1,
\end{equation}
where we have used that the two terms in the linear combination are
Hermitian conjugates of each other, hence their trace norms are equal.
This gives for the negativity 
\begin{equation}
\mathcal{N} \le \frac{1}{2}\left( \sqrt{2} {\rm tr} (O_+ O_-)^{1/2}-1\right), 
\label{nub}
\end{equation}
whereas the logarithmic negativity can be upper bounded as
\begin{equation}
\lneg \le \ln {\rm tr} (O_+ O_-)^{1/2} + \ln \sqrt{2} \, .
\label{lnub}
\end{equation}

The main advantage of these upper bounds is that they involve only
the product of Gaussian operators $O_+ O_-$, which is itself Gaussian
and the traces of its powers can be expressed via appropriate covariance matrix formulas.
To arrive to these expressions, it is useful first to introduce the normalized
Gaussian density operator
\begin{equation}
\rhox =  \frac{O_+O_-}{{\rm tr}(O_+O_-)},
\label{rhox}
\end{equation}
with corresponding covariance matrix $\Gx$.
The rules of multiplication are simplest to obtain by considering
the exponential form of the various Gaussian operators
\begin{equation}
\frac{1}{Z_\sigma}\exp \left( \sum_{k,l} (W_\sigma)_{k,l} m_k m_l/4 \right),
\label{expgauss}
\end{equation}
where the superscripts $\sigma=+,-$ and $\times$ refer to the corresponding
operator $O_+, O_-$ and $\rhox$. The matrices in the exponent
are related to the covariance matrices via
\begin{equation}
i \tanh \frac{W_\sigma}{2} = \gamma_\sigma, \qquad
\exp (W_\sigma) = \frac{1-i\gamma_\sigma}{1+i\gamma_\sigma},
\label{wg}
\end{equation}
and the normalization factors are given by
\begin{equation}
Z_\sigma = \dett (\1 + \exp(W_\sigma)).
\label{z}
\end{equation}
Here the symbol $\dett$ denotes that the double degenerate eigenvalues
of the corresponding matrix have to be counted only once, i.e. it is the
square root of the determinant up to a possible sign factor.
Using Eqs. \eqref{expgauss} and \eqref{wg}, the solution
for $\Gx$ can be found after simple algebra as \cite{FC10}
\begin{equation}
-i\Gx = \1 - 
(\1+i\Gm)(\1-\Gp\Gm)^{-1}(\1+i\Gp).
\label{gx}
\end{equation}
With the multiplication rule at hand, we are now ready to evaluate the trace norm
\begin{equation}
\|O_+\|_1 = {\rm tr} (O_+ O_-)^{1/2} = 
{\rm tr} (\rhox)^{1/2} \left(\frac{Z_\times}{Z_+Z_-}\right)^{1/2}
\label{trnop}
\end{equation}
appearing in the upper bounds \eqref{nub} and \eqref{lnub}. Using
\eqref{wg} and \eqref{z}, the ratio of the normalization factors can be rewritten as
\begin{equation}
\frac{Z_\times}{Z_+Z_-}= 
\dett \frac{\1-\Gp\Gm}{2} = \dett \frac{\1-\gamma^2}{2}.
\label{zratio}
\end{equation}
For the other term we can use the well-known trace formula for Gaussian states
\begin{equation}
{\rm tr} \rhox^\alpha = \dett \left[ \left(\frac{\1+i\Gx}{2}\right)^\alpha
+ \left(\frac{\1-i\Gx}{2}\right)^\alpha \right],
\label{trrhox}
\end{equation}
with $\alpha=1/2$. Hence, the upper bounds can be calculated explicitly
in terms of the covariance matrices $\Gx$ and $\gamma$.

Before moving to the study of concrete examples, let us comment about
the spectral properties of $\Gx$. By a similarity transformation
one can permute the factors in the second term of \eqref{gx} to arrive at
\begin{equation}
\Gx \simeq (\1-\Gp\Gm)^{-1}(\Gp + \Gm),
\label{gx2}
\end{equation}
where $\simeq$ denotes equivalence of the spectra. Furthermore, using
the definition in Eq.\ \eqref{gpm}, one can write
\begin{equation}
\Gx \simeq \left(\frac{\1-\gamma^2}{2}\right)^{-1} \frac{\gamma R + R \gamma}{2}, \qquad
\label{gx3}
\end{equation}
where $R=(T^+_B)^2=(T^-_B)^2=\1_{2|A|} \oplus -\1_{2|B|}.$ Thus the
second term in \eqref{gx3} becomes block diagonal
\begin{equation}
\frac{\gamma R + R \gamma}{2} =
\gamma_A \oplus -\gamma_B,
\end{equation}
with the sign of the reduced covariance matrix of the $B$ modes being reversed.
In particular, if the state on $A\cup B$ is pure, i.e. $\gamma^2=-\1$, then
the spectrum of $\Gx$ is simply given by the eigenvalues of $\gamma_{A}$
and $-\gamma_{B}$, respectively. Moreover, since the spectrum of $\gamma_{A}$
and $\gamma_B$ are identical (up to trivial eigenvalues $\pm i$ if $|A| \ne |B|$)
this just leads to a double degeneracy.

For the upper bound of the logarithmic negativity, it is useful to define the quantity
\begin{equation}
\lnub = \ln {\rm tr} (O_+ O_-)^{1/2}
\end{equation}
such that $\lneg \le \lnub + \ln \sqrt{2}$. Then using \eqref{trnop}-\eqref{trrhox},
$\lnub$ can be expressed via Renyi entropies as
\begin{equation}
\lnub = \frac{1}{2} \left[ S_{1/2}(\rhox) - S_2(\rho_{A\cup B}) \right],
\label{lnub2}
\end{equation}
where for any state $\rho$
\begin{equation}
S_\alpha(\rho) = \frac{1}{1-\alpha} \ln {\rm tr} \rho^\alpha.
\end{equation}
In particular, for pure states one has $S_2(\rho_{A\cup B})=0$, while
$S_{1/2}(\rhox)=2S_{1/2}(\rho_A)$ due to the double degeneracy of the
$\Gx$ spectrum mentioned above, and hence $\lneg = \lnub = S_{1/2}(\rho_A)$.
In other words, for pure states the upper bound is tight without the additional
constant $\ln\sqrt{2}$, since the operators $O_+$ and $O_-$ commute.

%
%
%
%
%
%


\section{Numerical examples\label{sec:num}}

In this section we will test the covariance-matrix based bounds introduced
before on the concrete example of a dimerized XX chain.
After Jordan-Wigner transformation, this is equivalent to a non-interacting fermionic 
chain with an alternating hopping $t_\pm=1\pm\delta$, given by the Hamiltonian
\begin{equation}
H = -\frac{1}{2} \sum_{j} \left( t_+ \, f_{2j}^\dag f_{2j-1}
+ t_- \, f_{2j+1}^\dag f_{2j} + \mathrm{h.c.} \right),
\label{ham}
\end{equation}
with dimerization parameter $-1 \le \delta \le 1$. 
This is also called the SSH chain. In all our examples we consider
an open chain with even sites $N$ at half filling, and calculate
the entanglement between the modes of two adjacent intervals, such that the
spin- and fermion-chain negativity are indeed equivalent.

A further simplification occurs due to the fact, that the Hamiltonian is particle-number conserving. 
On one hand, this allows us to implement our simple construction
for the lower bound. On the other hand, it makes the calculations for the upper bound easier,
since all the information is encoded in the fermionic correlation matrix elements
$C_{m,n}=\langle f^\dag_m f_n\rangle$. As already noted in \cite{EntanglementNegativityIn2D},
for a particle-conserving Gaussian state with real $C$ one can replace the covariance 
matrix $-i\gamma \mapsto G = 2C-1$ and define the
matrices $G_\pm$ and $G_\times$ correspondingly. The formulas leading to the upper bound
are then completely analogous to \eqref{zratio} and \eqref{trrhox}, except that the $\dett$ symbols have
to be replaced by ordinary determinants.

\subsection{Bounds vs. exact results}

First, we test both lower and upper bounds against exact calculations
of the logarithmic negativity for small chain sizes $N\le10$. For simplicity,
we consider two adjacent intervals of the same size $\ell$, taken symmetrically 
from the center of the chain. We will consider both ground and thermal states
of the dimerized chain, for which the fermionic correlation matrix elements read
\begin{equation}
C_{m,n} = \sum_{k=1}^{N} \frac{\phi^*_k(m)\phi_k(n)}{e^{\beta\omega_k}+1},
\end{equation}
where $\omega_k$ and $\phi_k(m)$ are the single-particle eigenvalues
and eigenvectors of the Hamiltonian \eqref{ham}.

Before presenting our data, let us comment on an observation about
the upper bound. Although the inequality reads $\lneg \le \lnub + \ln \sqrt{2}$,
in all our numerics we observe that the bound is actually tighter, i.e.
one has $\lnlb \le \lneg \le \lnub$. This has also been conjectured in
Ref.\ \cite{EstimatingNegativityForFreeFermions} but a rigorous proof is lacking.

In our first example we consider the ground state of a chain with $N=8$
and $\ell=2$. The data is shown in Fig.\ \ref{fig:lnbgs} as a function of the dimerization
parameter $\delta$. Note that, since $N/2=4$ is even, the hopping between
the two subsystems is given by $1-\delta$. Thus the entanglement vanishes
for $\delta=1$ while it is given by $\ln 2$ at the other extreme $\delta=-1$,
where a singlet is formed in the center. As expected from its construction,
the lower bound $\lnlb$ performs well only in the region $\delta<0$, where one
has a singlet-type dominant contribution to the entanglement.
Remarkably, the upper bound $\lnub$ gives an overall good performance on
both sides, with an almost perfect saturation for $\delta>0.2$. However, approaching
$\delta \to -1$, the entanglement tends to stay closer to its lower bound.

%
\begin{figure}[htb]
\center
\includegraphics[width=0.9\columnwidth]{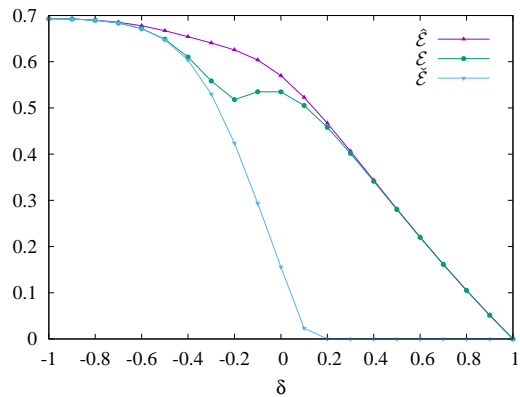}
\caption{Logarithmic negativity bounds vs. exact results in the ground state,
as a function of the dimerization $\delta$, with $N=8$ and $\ell=2$.}
\label{fig:lnbgs}
\end{figure}
%

It is very instructive to have a look also at the thermal case. Here we consider
the two halves of a chain with $N=8$ sites as subsystems and vary the temperature.
This scenario exhibits a very rich physics, as depicted on Fig.\ \ref{fig:lnbth}, where
now the symbols show the exact data, whereas the solid lines with matching colors
give the respective bounds. In fact, in the regime $\delta<0$ where the couplings at the boundaries are weak,
the Hamiltonian \eqref{ham} supports edge states. Consequently, the ground state shows
topological features which yields an additional $\ln 2$ contribution to the entanglement
as $\delta \to -1$. Since the state is pure, one has $\lneg=\lnub$, as discussed earlier.
However, already a slight increase of the temperature (see $\beta=100$) seems to
destroy this order, hence the topological contribution to the entanglement vanishes.
Not surprisingly, for these low temperatures the upper bound gives a very good
overall estimation. Nevertheless, for increasing temperatures, the data gradually
moves towards the lower bound. This improved performance can be understood
by a simple argument. The construction of the lower bound erases all the 
correlations within each subsystem $A$ and $B$. At higher temperatures, however,
such correlations are already washed out and thus the approximation is more valid.

%
\begin{figure}[htb]
\center
\includegraphics[width=0.9\columnwidth]{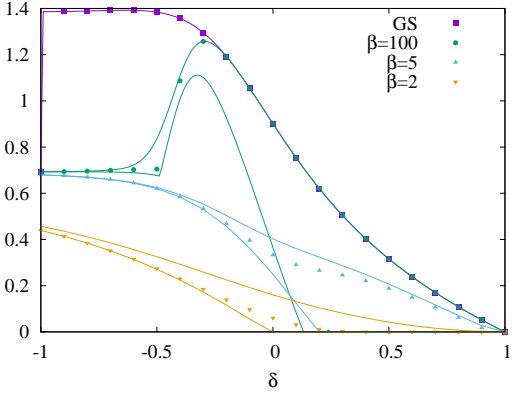}
\caption{Logarithmic negativity bounds vs. exact results for thermal states,
as a function of the dimerization $\delta$, and for various values of $\beta$.
The symbols represent the exact data, while the solid lines with matching colors
show the corresponding bounds.}
\label{fig:lnbth}
\end{figure}
%

\subsection{Upper bound for infinite homogeneous chain}

From now on we focus on the homogeneous chain $\delta=0$, and take
the thermodynamic limit $N \to \infty$. The Hamiltonian is
then diagonalized by a Fourier transform and the correlation matrix takes
the simple form
\begin{equation}
C_{m,n} = \int_{-\pi}^{\pi}\frac{d q}{2\pi} \frac{e^{-i(m-n)}}{e^{-\beta \cos q}+1}.
\end{equation}
Our main goal is to study the scaling of the upper bound as a function
of the inverse temperature $\beta$ and subsystem sizes $|A|=\ell_1$
and $|B|=\ell_2$ and compare it to the predictions of CFT \cite{QFTNegativity2}.

\subsubsection{Ground state}

We start with the study of $\lnub$ in the ground state and take $\ell_1=\ell_2=\ell$ for simplicity.
Invoking Eq.\ \eqref{lnub2}, one observes that the upper bound can be written as the difference of
two R\'enyi entropies, with respect to Gaussian states $\rho_{A\cup B}$ and $\rhox$.
Note that, while the former is just the reduced density operator of an interval of size $2\ell$ in
an infinite hopping chain, the latter one has no particular physical interpretation.

To understand the scaling behaviour of the entropies, it is useful to have a look
at the corresponding free-fermion entanglement Hamiltonians $\mathcal{H}$ and
$\mathcal{H}_\times$, defined by \cite{PE09}
\begin{equation}
\rho_{A \cup B} = \frac{e^{-\mathcal{H}}}{\mathcal{Z}}, \qquad
\rhox = \frac{e^{-\mathcal{H}_\times}}{\mathcal{Z}_\times}.
\label{entham}
\end{equation}
Their single-particle spectra, $\ek$ and $\exk$, respectively, are related via
\begin{equation}
\zeta_k = \tanh \frac{\ek}{2}, \qquad
\zeta^\times_k = \tanh \frac{\exk}{2}
\end{equation}
to the spectra $\zeta_k$ of $G$ and $\zeta^\times_k$ of $G_\times$.
Owing to the simple thermal form \eqref{entham} of the density operators,
the calculation of Renyi entropies reduces to evaluating entropy formulas for a Fermi gas.
In fact, the leading contributions to the entropies are delivered by the low-lying
eigenvalues of the spectra. For the entanglement Hamiltonian $\mathcal{H}$,
these were studied before and, for $\ln \ell \gg 1$, are given approximately by
\cite{Peschel04,EP13}
\begin{equation}
\ek = \frac{\pi^2(k-1/2-\ell)}{\ln(4\ell)-\psi(1/2)},
\label{epsk}
\end{equation}
with the digamma function $\psi(1/2)\approx -1.963$.
Thus the entanglement Hamiltonian has a level spacing inversely proportional
to $\ln \ell$, or in other words, a logarithmic density of states.
In turn, this yields the celebrated result for the Renyi entropies
\begin{equation}
S_\alpha(\rho_{A \cup B}) = \frac{1}{6}(1+\alpha^{-1}) \ln \ell + \mathrm{const.}
\label{entl}\end{equation}

We shall now have a look at the spectra $\exk$ and their behaviour as
a function of $\ell$, shown in Fig.\ \ref{fig:exl}. Apart from the double degeneracy
of the eigenvalues, the spectra show very similar features to those of $\ek$.
In particular, one can observe the slow logarithmic variation of the spacing
and the approximate linear behaviour around zero. We thus propose the ansatz
\begin{equation}
\varepsilon^\times_{2k-1}=\varepsilon^\times_{2k}
= a\frac{\pi^2(k-1/2-\ell/2)}{\ln(2\ell)+b},
\label{epsxk}
\end{equation}
with fitting parameters $a$ and $b$. Fitting the lowest-lying eigenvalue as a
function of $\ell$, we obtain $a=1.325\approx 4/3$ and $b=1.655$ where,
for better fit results, we also included a subleading term proportional to $1/\ell$.
Note that the higher part of the spectrum shows a slight upward bend which
is again very similar to the behaviour of the $\ek$ spectra \cite{EP13}.

%
\begin{figure}[htb]
\center
\includegraphics[width=0.9\columnwidth]{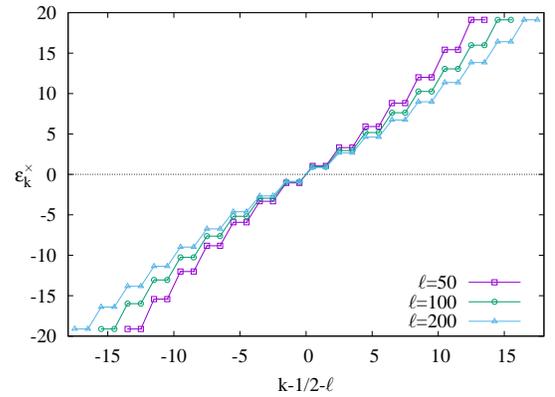}
\caption{Single-particle spectra $\exk$ for various $\ell$.}
\label{fig:exl}
\end{figure}
%

From the ansatz in Eq.\ \eqref{epsxk} it is very easy to infer the leading
scaling behaviour of the Renyi entropies. Indeed, the main difference from
\eqref{epsk} is the increased level spacing, leading to a decrease of the
density of states by a factor of $a^{-1} \approx 3/4$. Taking into account
also the double degeneracy of the spectrum, one arrives at
\begin{equation}
S_\alpha(\rhox) = \frac{1}{4}(1+\alpha^{-1}) \ln \ell + \mathrm{const.}
\label{snrhox}\end{equation}
That is, ignoring the subleading constant which is also modified due to
the parameter $b$, the entropies $S_\alpha(\rhox)$ and $S_\alpha(\rho_{A\cup B})$
differ by a factor of $3/2$. This is indeed the result we find numerically by
fitting the data for various $\alpha$. Finally, inserting the appropriate Renyi
entropies into \eqref{lnub2}, one immediately finds
\begin{equation}
\lnub = \frac{1}{4} \ln \ell + \mathrm{const.}
\label{lnubadj}
\end{equation}
Thus, the upper bound shows exactly the same scaling as the logarithmic
negativity predicted by CFT calculations with central charge $c=1$ \cite{QFTNegativity2}.

%
\begin{figure}[htb]
\center
\includegraphics[width=0.9\columnwidth]{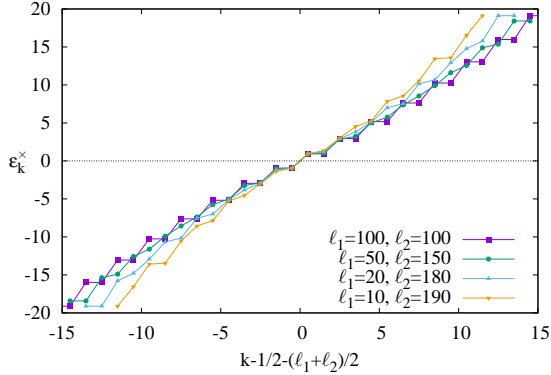}
\caption{Single-particle spectra $\exk$ for $\ell_1+\ell_2=200$ and various $\ell_1$.}
\label{fig:exl1l2}
\end{figure}
%

It is instructive to have a look also at the case of unequal adjacent segments
of size $\ell_1$ and $\ell_2$, where the CFT prediction gives \cite{QFTNegativity2}
\begin{equation}
\lneg=\frac{c}{4} \ln \frac{\ell_1 \ell_2}{\ell_1+\ell_2} + \mathrm{const.}
\label{lnl1l2}
\end{equation}
The corresponding spectra $\exk$ are shown in Fig.\ \ref{fig:exl1l2}, for a fixed overall
length $\ell_1+\ell_2=200$ and varying $\ell_1$. The main feature to be seen is the
breaking of the degeneracies. Indeed, from the analog of Eq.\ \eqref{gx3} to the present case,
it is clear that the spectrum of $G_\times$ must somehow mix those of $G_A$, $G_B$
and $G$, which is reflected on the corresponding single-particle entanglement
spectra. Unfortunately, however, it is very difficult to separate the various contributions
and, in contrast to the case of a single length scale in \eqref{epsxk},
we have not been able to find a simple ansatz. Nevertheless, from evaluating $\lnub$,
we find exactly the same scaling behaviour \eqref{lnl1l2} as obtained from CFT.
The results are plotted against the proper scaling variable on Fig.\ \ref{fig:lnubl1l2},
finding a perfect collapse of the data. Furthermore, comparing to the result for equal intervals
as a function of the segment size, we observe that the two functions match perfectly.

%
\begin{figure}[htb]
\center
\includegraphics[width=0.9\columnwidth]{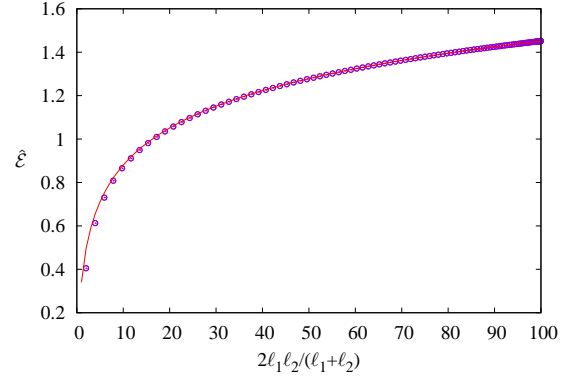}
\caption{Upper bound against CFT scaling variable with $\ell_1+\ell_2=200$
fixed and varying $\ell_1$. For comparison, the solid line shows the equal-segment
result \eqref{lnubadj}, with  $\ell_1=\ell_2=\ell$.}
\label{fig:lnubl1l2}
\end{figure}

\subsubsection{Thermal states}

As our final exmaple, we consider thermal states of the
infinite hopping chain with adjacent equal-size segments, where
the CFT calculation of the logarithmic negativity gives \cite{NegativityNonEq}
\begin{equation}
\lneg = \frac{c}{4} \ln \frac{\beta}{\pi} \tanh \frac{\ell \pi}{\beta} + \mathrm{const.}
\label{lnbeta}
\end{equation}
Hence, for any finite temperatures and $\ell \gg \beta$, the negativity
satifies an area law. To compare it to the behaviour of the upper bound,
one should first have a look at the corresponding spectra $\exk$,
shown in Fig.\ \ref{fig:exth} as a function of $\ell$ and for various $\beta$.
One sees the thermal flattening of the spectra with increasing temperatures,
which signals a crossover from logarithmic to linear density of states in $\ell$.

%
\begin{figure}[htb]
\center
\includegraphics[width=0.9\columnwidth]{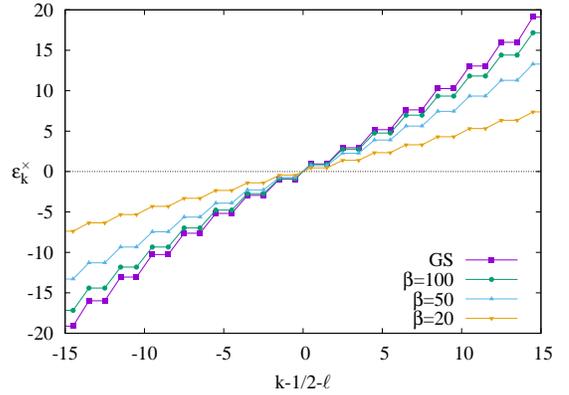}
\caption{Thermal single-particle spectra for $\ell=100$ and various $\beta$.}
\label{fig:exth}
\end{figure}
%

As an immediate consequence, the Renyi entropy $S_{1/2}(\rhox)$
becomes extensive. This, however, does not necessarily spoil the tightness of
our upper bound, since the contribution from $S_{2}(\rho_{A\cup B})$, which is
itself extensive, has to be subtracted. Indeed, as shown in Fig.\ \ref{fig:lnubth},
we find numerically that $\lnub$ saturates for large $\ell$ for any nonzero
temperatures and hence the extensive contributions from the two entropies exactly cancel.
Moreover, as shown on the inset, we confirm that $\lnub$ has exactly the same scaling
behaviour as $\lneg$ in \eqref{lnbeta}. Note, however, that it is difficult to find an analytic
argument to understand this type of scaling on the level of the spectra $\exk$,
since one has to look for subleading effects.

%
\begin{figure}[htb]
\center
\includegraphics[width=0.9\columnwidth]{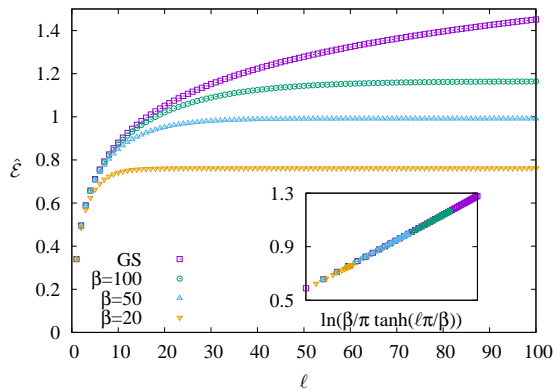}
\caption{Upper bound for thermal states with various $\beta$ against $\ell$.
The inset shows the data against CFT scaling variable.}
\label{fig:lnubth}
\end{figure}

\medskip
\section{Outlook\label{sec:outlook}}

In conclusion, we have presented rigorous bounds to
the entanglement negativity that are efficiently computable
for fermionic Gaussian states.
In particular, the definition of the lower bound and
one of the upper bounds is a simple function of the
covariance matrices, allowing an efficient calculation
in the number of fermionic modes. Furthermore, we
have also constructed an upper bound which makes use
of semi-definite programming techniques.

There are a number of questions left open for future research.
First, in all our numerical examples, carried out for adjacent
intervals in a dimerized hopping chain, we observed
that the upper bound of the logarithmic negativity can actually
be made more tight by neglecting an additive
constant $\ln \sqrt{2}$. Although it has also been conjectured in
Ref. \cite{EstimatingNegativityForFreeFermions},
we could not give a rigorous proof in support of this claim
and it is still unclear if this holds in complete generality.

Moreover, while the upper bound for adjacent intervals in a free-fermion
chain gives exactly the same scaling behaviour as the CFT prediction for
the entanglement negativity, one should also test its performance
for the case of non-adjacent intervals.
Unfortunately, this setup is much more involved since the analytic
continuation from the moments of the partial transpose is not known
\cite{TwoDisjointIntervalsNegativity}.
Another interesting question is the negativity for non-adjacent intervals
in the XX spin chain, where the results in the spin and fermionic basis are
not equivalent \cite{CTC15,CTC16},
and thus the upper bound should also be properly generalized.

Regarding the lower bound, we observed that it performs
particularly well in case of strong singlet-type entanglement between
the subsystems. This makes it a good candidate to check the
negativity scaling in random singlet phases of disordered spin chains,
where the available DMRG results are not yet entirely conclusive
\cite{RandomSpinChainNegativity}. Importantly, the bounds presented here constitute an excellent starting point
for endeavors aimed at seeing topological signatures at finite temperatures, as the 
numerics for comparably small SSH chains already suggests.
It is the hope that this work stimulates
such further research.\\


\noindent \emph{Note added.} Upon completion, 
we became aware of a recent independent work
\cite{PartialTimeReversalNegativity},
where an alternative definition of fermionic entanglement
negativity is considered. Making use of a freedom in the
representation of the partial transposition, the authors
adopt a different convention is equivalent to partial
time-reversal. In turn, their entanglement negativity
coincides with our upper bound $\lnub$ in Section VII.

\section{Acknowledgements}
We would like to thank discussions and correspondence with Hassan Shapourian, Shinsei Ryu, Ingo Peschel, Christopher Herzog, Yihong Wang, Vladimir Korepin, Erik Tonni, Andrea Coser, and Pasquale Calabrese. 
J.\ E.\ and Z.\ Z.\ have been supported by the DFG (CRC183, EI 519/9-1, EI 519/7-1), 
the Templeton Foundation, and the ERC (TAQ). Z.\ Z.\ would also like to thank the Simons Center for Geometry and Physics for hospitality where some of the work has been carried out.
V.\ E.\ acknowledges funding from the Austrian Science Fund (FWF) through Lise Meitner Project No. M1854-N36.


%

\end{document}